\documentclass{article}

\usepackage{amsthm}
\usepackage{amssymb,amsmath,amsfonts}
\usepackage{enumerate}
\usepackage{algorithm,algpseudocode}
\usepackage{mathtools}
\usepackage{graphicx}
\usepackage[usenames]{color}
\usepackage{xspace}
\usepackage{tikz,pgfplots}
\usepackage{fullpage}

\newcommand{\BN}{\mathbb{N}}
\newcommand{\BR}{\mathbb{R}}

\newtheorem{theorem}{Theorem}
\newtheorem{lemma}[theorem]{Lemma}
\newtheorem{corollary}[theorem]{Corollary}

\newtheorem{remark}[theorem]{Remark}

\title{Deterministic $O(1)$-Approximation Algorithms to 1-Center Clustering with Outliers}
\author{Shyam Narayanan\\
\small Harvard University\\[-0.8ex] 
\small\tt shyam.s.narayanan@gmail.com
}

\begin{document}

\maketitle

\begin{abstract}
    The {\em 1-center clustering with outliers} problem asks about identifying a prototypical robust statistic that approximates the location of a cluster of points.  Given some constant $0 < \alpha < 1$ and $n$ points such that $\alpha n$ of them are in some (unknown) ball of radius $r,$ the goal is to compute a ball of radius $O(r)$ that also contains $\alpha n$ points.  This problem can be formulated with the points in a normed vector space such as $\mathbb{R}^d$ or in a general metric space.
    
    The problem has a simple randomized solution: a randomly selected point is a correct solution with constant probability, and its correctness can be verified  in linear time. However, the {\em deterministic} complexity of this problem was not known. In this paper, for any $\ell_p$ {\em vector space}, we show an $O(nd)$-time solution with a ball of radius $O(r)$ for a fixed $\alpha > \frac{1}{2},$ and for any {\em normed vector space}, we show an $O(nd)$-time solution with a ball of radius $O(r)$ when $\alpha > \frac{1}{2}$ as well as an $O (nd \log^{(k)}(n))$-time solution with a ball of radius $O(r)$ for all $\alpha > 0, k \in \mathbb{N},$ where $\log^{(k)}(n)$ represents the $k$th iterated logarithm, assuming distance computation and vector space operations take $O(d)$ time.  For an {\em arbitrary metric space}, we show for any $C \in \mathbb{N}$ an $O(n^{1+1/C})$-time solution that finds a ball of radius $2Cr,$ assuming distance computation between any pair of points takes $O(1)$-time. Moreover, this algorithm is optimal for general metric spaces, as we show that for any fixed $\alpha, C,$ there is no $o(n^{1+1/C})$-query and thus no $o(n^{1+1/C})$-time solution that deterministically finds a ball of radius $2Cr$.
\end{abstract}
 
\section{Introduction}

Data clustering that is tolerant to outliers is a well-studied task in machine learning and computational statistics. 
In this paper, we deal with one of the simplest examples of this class of problems: {\em 1-center clustering with outliers}.  Informally, given $n$ points such that there exists an unknown ball of radius $r$ containing most of the points, we wish to find a ball of radius $O(r)$ also containing a large fraction of the points. More formally, suppose $0 < \alpha < 1$ is some fixed constant.  Given points $a_1, ..., a_n$ in space $\BR^d$ (where points are given as coordinates) under an $\ell_p$ norm for some $p \ge 1$, in some other normed vector space, or in an arbitrary metric space (where we just have access to distances), suppose we know there exists a ball of radius $r$ containing at least $\alpha n$ points but do not know the location of the ball.  Then, can we efficiently provide a $C$-approximation to finding the ball, i.e. find the center of a ball of radius $Cr$ for some $C \ge 1$ containing at least $\alpha n$ points?

The problem has a simple linear-time Las Vegas randomized algorithm: a randomly selected point is a correct solution with constant probability, and its correctness can be verified in linear time. In fact, an even faster randomized algorithm works by picking $O(1)$ points randomly, computing pairwise distances, and selecting a cluster if it exists. However, the {\em deterministic} complexity of this problem appears more intriguing, and to the best of our knowledge, no linear-time or even subquadratic-time (let alone simple) solution for this problem was known.  A trivial quadratic-time algorithm exists by enumerating over all points and checking pairwise distances, so the goal of the paper is to obtain deterministic algorithms whose running time is faster than the above. This situation bears similarity to the closely related {\em 1-median} problem, where given a set of points $a_1, ..., a_n$ we want to find a point $p^*$ that (approximately) minimizes the sum of the distances between $p^*$ and all $a_i$'s. It is a folklore fact that a randomly selected point  is a $2(1+\epsilon)$-approximate $1$-median with probability at least $\frac{\epsilon}{1+\epsilon}$. However, in the deterministic case for an arbitrary metric space, no constant-factor approximation in linear time is possible \cite{Lowerbound1, Lowerbound2}, and non-trivial tradeoffs between the approximation factor and the running time exist~\cite{Upperbound1, Upperbound2}. The goal of this paper is to establish an analogous understanding of the deterministic complexity of 1-center clustering with outliers.

\subsection{Main results}

Our results are depicted in Table~\ref{table}. They primarily fall into two main categories: results in normed vector spaces and results in arbitrary metric spaces.  For $\BR^d$ with the $\ell_p$ norm, assuming we are given coordinates of points, our algorithm runs in $O(nd)$ time with an $O((\alpha-0.5)^{-1/p})$-approximation, assuming $\alpha > \frac{1}{2}.$  Such a runtime even for the Euclidean case was previously unknown.  For arbitrary normed vector spaces, our algorithm runs in $O_{\alpha}(nd)$ time with an $O((\alpha-0.5)^{-1})$-approximation whenever $\alpha > 0.5,$ assuming that distance calculation, vector addition, and vector multiplication can be done in $O(d)$ time.  For $0 < \alpha \le 0.5,$ we solve the problem for arbitrary normed vector spaces in $O_{\alpha, k} (nd \log^{(k)}(n)) = O_{\alpha, k} (nd \log \log ... \log n)$ time for any integer $k$.

For arbitrary metric spaces, assuming distance calculation takes $O(1)$ time, we give an $O_{\alpha, C}(n^{1+1/C})$-time algorithm with approximation constant $2C.$  While this is much weaker than for normed vector spaces, this runtime is actually tight for a fixed $\alpha, C$, as there is no $o(n^{1+1/C})$-time algorithm with approximation constant $2C$ that works for an arbitrary metric space.  In particular, this implies there is no $O(n \text{ polylog } n)$-time solution to solve the general metric space problem, even for large $\alpha$ and $C$.

As a note, subscripts of $\alpha, k,$ and $C$ on our $O$ and $\Omega$ factors mean that the constants may depend on $\alpha, k,$ and $C,$ but are bounded by some function of the subset of $\alpha^{-1}, k, C$ in the subscript.

\vskip 0.6cm

\begin{table}[htbp]
\begin{center}
\begin{tabular}{ |p{1.7cm}|p{2cm}|p{2.75cm}|p{2.8cm}|p{3.45cm}| }
 \hline
 Space & Assumptions & Runtime & Approximation & Comments\\
 \hline
 \hline
$\ell_p$ normed & $\alpha > \frac{1}{2}$ & $O(nd)$ & $O\left((\alpha-0.5)^{-1/p}\right)$ & Implies Euclidean\\ \hline
Normed & $\alpha > \frac{1}{2}$ & $O_{\alpha}(nd)$ & $O\left((\alpha-0.5)^{-1}\right)$ & \\ \hline
Normed & $\alpha > 0$ & $O_{\alpha, k}(nd \log^{(k)}(n))$ & $O_{\alpha, k}(1)$ & Implies for $\ell_p$ space, \newline $k$ any positive integer \\ \hline
Metric & $\alpha > 0$ & $O_{\alpha, C}(n^{1+1/C})$ & $2C$ & Can be done even if the radius is unknown, \newline $C$ any positive integer\\ \hline
Metric & $\alpha > 0$ & $\Omega_{\alpha, C}(n^{1+1/C})$ & $2C$ & Adversary from metric $1$-median lower bound \\ \hline
\end{tabular}
\caption{Our results}
\label{table}
\end{center}
\end{table}
\vspace{-2em} 

\subsection{Motivation and Relation to Previous Work}

1-center clustering with outliers is a very simple example of a robust statistic, i.e. its location is usually resistant to large changes to a small fraction of the data points.  Robust statistics are reviewed in detail in \cite{RobustStatisticsBook}.  When $\alpha > \frac{1}{2},$ addition of a large number of points does not change the statistic up to $O(r),$ as it only slightly decreases the value of $\alpha.$  Even if $\alpha < \frac{1}{2},$ the statistic is still robust as if we find some ball containing $\alpha n$ points that are disjoint from the intended ball, we can remove those points and now there is some ball with at least $\alpha' = \frac{\alpha}{1-\alpha}$ of the remaining points which we need to get close to, so inducting on $\lfloor \alpha^{-1} \rfloor$ shows that the statistic is robust.

Robust statistics have a lot of practical use in statistics and machine learning \cite{CharikarSteinhardtValiant, DiakonikolasEtAl}.  Since machine learning often deals with large amounts of data, it is difficult to obtain a large amount of data with high accuracy in a short period of time.  Therefore, if we can compute a robust statistic quickly, we can get more data in the same amount of time and have a good understanding of the approximate location of a good fraction of the data.

This question is valuable from the perspective of derandomization.  One solution to the 1-center clustering problem is to randomly select a point and check if it is at most $2r$ away from $\alpha n - 1$ other points, and repeat the process if it fails.  This algorithm is efficient and gets a ball of radius $2r$ with $\alpha n$ points after $O(\alpha^{-1} n)$ expected computations, but is a Las Vegas algorithm that can be slow with reasonable probability.  A faster Monte Carlo algorithm involves choosing an $O(1)$-size subset of the points and running the brute force quadratic algorithm, though similarly this algorithm may fail with reasonable probability.  Therefore, this problem relates to the question of the extent to which randomness is required to solve certain computational problems.

The Euclidean problem is useful in the amplification of an Approximate Matrix Multiplication (AMM) algorithm described in \cite{ClarksonWoodruff2009}.  To compute $A^T B$ up to low Frobenius norm error with probability $2/3$ in low time and space, the algorithm approximates $A^T B$ as $C = (SA)^T (SB),$ where $S$ is a certain randomized sketch matrix.  Then, if this process is repeated $O(\log \delta^{-1})$ times to get $C_1, ..., C_{O(\log \delta^{-1})}$, with probability $1-\delta,$ at least $3/5$ of the $C_i$'s satisfy $||C_i - A^T B||_F \le \epsilon ||A||_F ||B||_F.$  We are able to approximate $||A||_F ||B||_F$ with high probability using $L_2$ approximation algorithms from \cite{AMS}.  If we think of $C_i$ and $A^T B$ as vectors, at least $3/5$ of them are in a ball of radius $r = \epsilon ||A||_F ||B||_F$ with probability $1-\delta$.  To approximate the center of this ball, i.e. $A^T B,$ they use the Las Vegas algorithm.  If we only assume that at least $3/5$ of the vectors are in a ball of radius $r,$ approximating the ball this way with probability $1-\delta$ requires $\Omega((\log \delta^{-1})^2)$ pairwise distance computations and thus $\Omega(d (\log \delta^{-1})^2)$ time where $d$ is the dimension of $A^T B$ as a vector.  However, Theorem \ref{EuclideanAboveHalf} gives a method that only requires $O(\log \delta^{-1})$ distance computations and $O(d \log \delta^{-1})$ time, thus making amplification of the error for this AMM algorithm linear in $\log \delta^{-1}.$

1-center clustering with outliers is also related to the standard 1-center problem (without outliers), which asks for a point $p$ that minimizes $\max_i \rho(p, a_i),$ where $\rho$ denotes distance \cite{1Center1}.  1-center with outliers has been studied, e.g., in \cite{Streaming1Center}, but under the assumption that the number of outliers is $o(n),$ instead of up to $(1-\alpha) n$.  The $1$-center and $1$-center with outliers problems also have extensions to $k$-center \cite{ChakrabartyEtAl} and $k$-center with outliers \cite{McCutchenEtAl, CharikarEtAl}, where there are up to $k$ allowed covering balls.  It also relates to the geometric 1-median approximation problem, which asks, for a set of points $a_1, ..., a_n,$ for some point $p^*$ such that
\[\sum\limits_{i = 1}^{n} \rho(p^*, a_i) \le C \cdot \min\limits_{p} \sum\limits_{i = 1}^{n} \rho(p, a_i),\]
i.e. finding a $C$-approximation to the geometric 1-median problem.
The geometric 1-median problem has been studied in detail, though usually focusing on randomized $(1+\epsilon)$-approximation algorithms in Euclidean space \cite{CohenEtAl, ChinCitation7}.  For the deterministic case, the centroid of all points in Euclidean space is known to be a $2$-approximation to geometric $1$-median, but in an arbitrary metric space, there exist tight upper \cite{Upperbound1, Upperbound2} and lower time bounds \cite{Lowerbound1, Lowerbound2} for all $C$.  The geometric 1-median problem is closely related to the 1-center clustering with outliers problem, and we show in Section \ref{SectionMetricUpper} that the adversary used for the proof of the lower bound for geometric $1$-median can establish an analogous lower bound for $1$-center clustering with outliers.  We thus establish tight upper and lower bounds for 1-center clustering with outliers in general metric space.

As a remark, our Theorem \ref{NormedVectorSpaces} uses an idea of deleting points that are far apart from each other, which is similar to certain ideas for insertion-only $\ell_1$-heavy hitters algorithms by Boyer and Moore and by Misra and Gries \cite{BoyerMoore, MisraGries}, in which seeing many distinct elements results in a similar deletion process.

\subsection{Notation}

For many of our proofs, we deal with a weighted generalization of the problem, defined as follows.  Let $\alpha$ and $a_1, ..., a_n$ be as in the original problem statement, but now suppose each $a_i$ has some weight $w_i \ge 0$ such that $w_1+...+w_n > 0.$  Furthermore, assume there is a ball of radius $r$ containing some points $a_{i_1},...,a_{i_s}$ such that $w_{i_1}+...+w_{i_s} \ge \alpha(w_1+...+w_n).$   The goal is then to find  a ball of radius $O(r)$ containing points $a_{j_1}, ..., a_{j_t}$ such that $w_{j_1}+...+w_{j_t} \ge \alpha (w_1+...+w_n)$, which we call containing at least $\alpha(w_1+...+w_n)$ weight.

Given points $a_1, ..., a_n$ with weights $w_1, ..., w_n,$ we let $w = \sum_{1 \le i \le n} w_i$, i.e. the total weight.  For any set $S \subset [n],$ let $a_S = \{a_i: i \in S\}$ and let $w_S = \sum_{i \in S} w_i$.  For some results, we define a new set of points $q_1, ..., q_m$ with weights $v_1, ..., v_m,$ so we will use the terms ``$w$-weight'' and ``$v$-weight'' accordingly if necessary.  Similarly for any set $S \subset [m],$ let $q_S = \{q_i: i \in S\}$ and let $v_S = \sum_{i \in S} v_i.$

For computing distances, $||x-y||$ denotes distance in a normed vector space, and $\rho(x, y)$ denotes distance in an arbitrary metric space.

Since $\alpha,$ the fraction of points or weight in the ball of radius $r$, is variable, we define the problem $1$-center clustering with approximation constant $C$ and fraction $\alpha$ as the problem where if there is a ball of radius $r$ containing $\alpha n$ points (or $\alpha w$ weight), we wish to explicitly find a ball of radius $Cr$ with the same property.

Finally, for any function $f$ in this paper, the following assumptions are implicit: $f$ is nondecreasing, $f(n) \ge 1,$ $f(n) = O(n),$ and $f(an) \le a f(n)$ for any $a \ge 1, n \in \BN$.

\subsection{Proof Ideas}

While many of our proofs assume the weighted problem, we assume the unweighted problem here for simplicity.  

The algorithm for the $\ell_p$ normed vector space simply returns the point whose $i$th coordinate is the median of the $i$th coordinates of $a_{[n]}$.  The proof is done shortly and is quite brief, so it is not included in this section.  We now describe the algorithm intuition for normed vector spaces when $\alpha > \frac{1}{2}.$  Our goal is to reduce the $n$ point problem into an $n/2$ point problem in $O_{\alpha}(nd)$ time, which means the overall runtime is $O_{\alpha}(nd)$.  To do this, we divide the $n$ points into $n/2$ pairs of points just by grouping the first two points, then the next two, and so on.  The idea is that when two points are far away, i.e. more than $2r$ apart, at most one of them can actually be in our ball $B$, so deleting both of them still means at least $\alpha$ of the points are in the ball of radius $r$.  However, when the two points are within $2r$ of each other, we ``join'' the points by pretending the second point is at the location of the first point, though as a result now we are only guaranteed a ball of radius $3r$ concentric with $B$ having $\alpha$ of the points, because we may join a point in the ball with a point close to the ball but not in it.  This means if we have a $C$ approximation for $n/2$ points, we can get a $3C$-approximation for $n$ points, since every remaining pair has the points in the same location so we keep only one point from each pair.  However, to go from a ball of radius $3Cr$ to a ball of radius $Cr,$ we look at the original set of points and take the centroid of all the points in the ball of radius $3Cr$.  The ball of radius $r$ containing at least $\alpha n$ points will cause the centroid to move closer to the ball, assuming $C$ is not too small.  We may have to repeat the process several times with smaller balls until we get sufficiently close, i.e. back to less than $Cr$ away from $B$, but this only requires $O_{\alpha}(1)$ iterations and thus $O_{\alpha}(nd)$ total time.

Unfortunately, for normed vector spaces when $\alpha \le \frac{1}{2},$ the centroid of the points within a certain radius may not be closer to the desired ball.
The idea to fix this is to assume that $B$ has at least $\alpha n$ more points than $B^C \backslash B$ for a certain constant $C,$ where for any $A$, $B^A$ is the ball of radius $Ar$ concentric with $B$.  Then, if we split the points into two halves, at least one half satisfies the same property.  Suppose that given $n/2$ points with this property we can find a ball of radius $Kr$ that not only contains at least $\alpha n$ points but also intersects $B$, for some $K \le \frac{C-3}{2}.$  Then, the ball of radius $(K+2)r$ around one of these points contains $B$ but is contained in $B^{C},$ so if we restrict to the ball of radius $(K+2)r$ around that point, at least $\frac{1+\alpha}{2}$ of the remaining points are in $B$, which has $\alpha n$ points.  Now, use the previous algorithm with some constant which is at least $\frac{1+\alpha}{2} > \frac{1}{2}$ to find a ball of radius $Kr$ with $\alpha n$ points, where we make sure $K$ is not too small.  However, there is an issue of multiple completely disjoint balls of radius $O(r)$, each having at least $\alpha n$ points, as $\alpha < \frac{1}{2}.$  To salvage this, we have to first find a ball of radius $Kr$ containing $\alpha n$ points, then remove the points in the ball and repeat the procedure with a higher value of $\alpha$, in case the ball we found does not actually intersect $B$.  Overall, this happens to make the runtime $O(nd \text{ polylog } n).$  One issue is that we don't know whether there is some $B$ that contains at least $\alpha n$ more points than $B^C \backslash B,$ but if there were some $B$ of radius $r$ that contains at least $\alpha n$ total points, for some $b = O(\log \alpha^{-1})$, $B^{C^b}$ contains at least $\frac{\alpha}{2} n$ more points than $B^{C^{b+1}} \backslash B^{C^{b}},$ or else the number of points would become too large.  Therefore, we attempt the procedure with fraction $\frac{\alpha}{2}$ for radius $r$, radius $Cr,$ radius $C^2r,$ and so on until $C^{O(\log \alpha^{-1})} r.$  Finally, we can go from $nd \text{ polylog } n$ to $nd \log^{(k)} n$ using a brute force divide and conquer.  Namely, if we can solve the problem in time $nd f(n),$ split the points into buckets of size $f(n),$ run the algorithm on each bucket, perhaps with a smaller value of $\alpha,$ and return $O(\frac{n}{f(n)})$ points in time $O(nd f(f(n)))$.  If we choose the points well, we get that most of the chosen points will be at most $Cr$ away from our desired ball $B$, so with a larger constant on the order of $C^2$, we can run the algorithm on the $O(\frac{n}{f(n)})$ points, which takes $O(nd)$ time.  We can repeat the procedure to get $O(nd f^{(k)}(n))$ for any $k$, though $C$ may become very large.

Our metric space bound ideas are almost identical in the cases of $\alpha > \frac{1}{2}$ and $\alpha \le \frac{1}{2},$ except for the issue that when $\alpha \le \frac{1}{2},$ we run into issues of finding a ball of radius $Cr$ with $\alpha n$ points that isn't near the desired ball of radius $r$ and $\alpha n$ points.  This issue is fixed by ideas of removing the points in the ball of radius $Cr$ and retrying the algorithm for a larger value of $\alpha$ if necessary.  For simplicity we assume $\alpha > \frac{1}{2}$.

For metric space upper bounds, one can use brute force divide and conquer.  Suppose in time $O(n^{1+1/K})$ we can solve the problem with approximation constant $C$.  Then, split the $n$ points into blocks of size $n^{K/(K+1)}$.  If we let the $i$th block be called $D_i,$ then some block must have at least $\alpha |D_i|$ points.  Therefore, if we run the algorithm on all blocks, which takes $O(n \cdot (n^{K/(K+1)})^{1/K}) = O(n^{1+1/(K+1)})$ time, for at least one block we will get a point at most $Cr$ away from $B$, which means the ball of radius $(C+2)r$ from some point must contain $B$ and thus at least $\alpha n$ total points.  There are $O(n^{1/(K+1)})$ points we have to check, each of which takes $O(n)$ time to verify, so we will find a point such that the ball of radius $(C+2)r$ contains at least $\alpha n$ total points in $O(n^{1+1/(K+1)})$ time.  As $\alpha > \frac{1}{2},$ this ball by default intersects any ball of radius $r$ with at least $\alpha n$ points.  Therefore, if we can solve the problem with approximation constant $C$ in $O(n^{1+1/K})$ time, we can solve the problem with constant $C+2$ in time $O(n^{1+1/(K+1)}),$ since the divide and conquer procedure and checking both take $O(n^{1+1/(K+1)})$ time.  Since a $2$-approximation in $n^2$ time is trivial, this should give a $2C$ approximation in $O(n^{1+1/C})$ time.  See Figure \ref{MetricSpacePicture} for an example when $C = 2.$

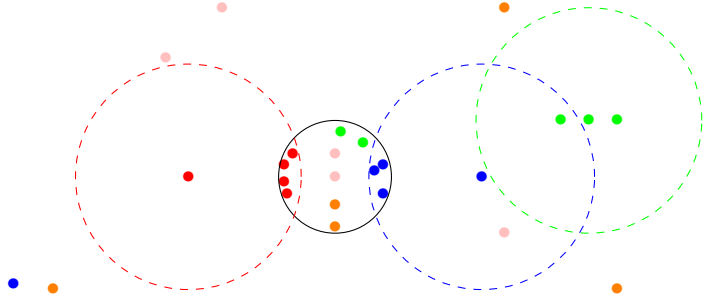
\begin{figure}
\centering
\begin{tikzpicture}[x=0.75cm,y=0.75cm]
\foreach \Point in {(-0.9, 0.2), (-0.75, 0.4), (-0.85, -0.3), (-0.9, -0.1), (-2.6, 0)}{
    \node [red] at \Point {\textbullet};
}
\foreach \Point in {(0.85, 0.2), (0.85, -0.3), (0.7, 0.1), (2.6, 0), (-5.7, -1.9)}{
    \node [blue] at \Point {\textbullet};
}
\foreach \Point in {(0.1, 0.8), (0.5, 0.6), (5, 1.0), (4, 1.0), (4.5, 1.0)}{
    \node [green] at \Point {\textbullet};
}
\foreach \Point in {(0, 0), (0, 0.4), (-2, 3), (-3, 2.1), (3, -1)}{
    \node [pink] at \Point {\textbullet};
}
\foreach \Point in {(0, -0.9), (0, -0.5), (3, 3), (-5, -2), (5, -2)}{
    \node [orange] at \Point {\textbullet};
}

\draw (0, 0) circle (0.75cm);

\draw[red, dashed] (-2.6, 0) circle (1.5cm);
\draw[blue, dashed] (2.6, 0) circle (1.5cm);
\draw[green, dashed] (4.5, 1) circle (1.5cm);
\end{tikzpicture}
\caption{Here is an example for $n = 25, \alpha = 13/25 = 0.52,$ and $C = 2.$  We split the $n = 25$ points into $\sqrt{n} = 5$ buckets of $\sqrt{n} = 5$ points each, color coded red, blue, green, pink, and orange.  The black circle represents the desired ball $B$ of radius $r$.  By brute force we try to find a ball of radius $2r$ containing at least an $\alpha$ fraction for each color, and succeed for red, blue, and green (represented by dashed circles).  It takes $O(\sqrt{n}^2) = O(n)$ time to try for each color, so the total time for this is $O(n \sqrt{n}).$  However, at least an $\alpha$ fraction of points of some color (in this case red) must be in $B$ by Pigeonhole, so the brute force algorithm must succeed in finding a ball of radius $2r$ containing an $\alpha$ fraction of the red points, and since $\alpha > 1/2,$ the radius $2r$ ball must contain some point in $B$ and thus must intersect $B$.  This means the ball of radius $4r$ concentric with the dashed red circle must contain $B$ by the triangle inequality, and thus has at least $\alpha n$ points.  We can check this for any ball in $O(n)$ time and there are at most $\sqrt{n}$ balls to check, so the total time for this is $O(n \sqrt{n})$.}
\label{MetricSpacePicture}
\end{figure}

The metric space lower bound comes from a lower bound by Chang on geometric $1$-median in metric space \cite{Lowerbound1}.  The paper by Chang constructs an algorithm for an adversary $Adv$ such that for any algorithm $A$ that only uses $o_{C}(n^{1+1/C})$ queries distances between pairs of vertices $a_i$ and $a_j$ where $A$ finally returns some point $a_z$, the adversary can adaptively create a metric space.  This metric space has distances $a_1, ..., a_n$ satisfy all metric space properties, returns correct distances between any two points queried, and such that $a_z$ is not a $2C$-approximation to geometric $1$-median.  However, analysis provided by Chang can be used to prove that this point $a_z$ is also not a $2C$-approximation for $1$-center clustering with outliers.  As a result, we not only get an $\Omega(n^{1+1/C})$-time lower bound for $2C$-approximation, but also a stronger $\Omega(n^{1+1/C})$-query lower bound.

\section{Normed Vector Space Algorithms: $\alpha > 1/2$} \label{NormedVectorSpaces}

For $\ell_p$ norms over $\BR^d$, there exists a straightforward algorithm.  Assume we are given the points $a_1, ..., a_n$ with weights $w_1, ..., w_n$ such that $(a_j)_i$ is the $i$th coordinate of $a_j.$  Then, consider the point $x = (x_1, ..., x_d)$ such that $x_i$ is the weighted median of $(a_1)_i, ..., (a_n)_i$ where $(a_j)_i$ has weight $w_j.$  Weighted median finding is known to take $O(n)$ time, so $x$ can be found in $O(nd)$ time.  Clearly, if there is a ball of radius $r$ around some $q$ with $\alpha w$ weight, where $\alpha > \frac{1}{2},$ then clearly $|q_i-x_i| \le r$ for each $i$, so $||q-x||_p \le r \cdot d^{1/p}.$  However, we can actually get another more valuable bound.

\begin{theorem} \label{MedianOfEachCoordinate}
    If $q$ is a point such that $B(q),$ the $\ell_p$-norm ball of radius $r$ around $q$, contains $\alpha w$ weight for some $\alpha > \frac{1}{2},$ then $||x-q||_p \le \left(\frac{\alpha}{\alpha-1/2}\right)^{1/p} r,$ implying an $O(nd)$ time solution with fraction $\alpha$ and approximation constant $O((\alpha-1/2)^{-1/p})$.
\end{theorem}

\begin{proof}
    Let $(q_1, ..., q_d)$ be the coordinate representation of $q$, and assume WLOG that $q_i \le x_i$ for each $1 \le i \le d$.  Suppose $B(q)$ contains exactly $\beta w$ weight, where $\beta \ge \alpha.$  If we let $(a_j)_{i}$ denote the $i$th coordinate of point $a_j,$ the set of points in $\{a_1, ..., a_n\} \cap B(q)$ with $(a_j)_{i} \ge x_i$ have at least $(\beta-1/2) w$ weight, as $x_i$ is the weighted median of the $i$th coordinate of all $n$ points.  Therefore,
\[\sum\limits_{a_j \in B(q)} w_j |(a_j)_i-q_i|^p \ge \left(\beta - \frac{1}{2}\right) w (x_i-q_i)^p\]
    for each $i$, meaning that if we sum over all $i$,
\[\sum\limits_{a_j \in B(q)} w_j ||a_j-q||_p^p = \sum\limits_{i = 1}^{d} \sum\limits_{a_j \in B(q)} w_j |(a_j)_i-q_i|^p\]
\[\ge \sum\limits_{i = 1}^{d} \left(\beta - \frac{1}{2}\right) w (x_i-q_i)^p = \left(\beta - \frac{1}{2}\right) w ||x-q||_p^p.\]
    However, $a_j \in B(q)$ means $||a_j-q||_p^p \le r^p,$ and as the weight of points in $B(q)$ equals $\beta w,$ 
\[(\beta w) \cdot r^p \ge \sum\limits_{a_j \in B(q)} w_j ||a_j-q||_p^p \ge \left(\beta - \frac{1}{2}\right) w ||x-q||_p^p\]
    which implies that
\[||x-q||_p \le \left(\frac{\beta}{\beta - \frac{1}{2}}\right)^{1/p} r \le \left(\frac{\alpha}{\alpha - \frac{1}{2}}\right)^{1/p}r.\]
    Thus, the ball of radius $\left(\left(\frac{\alpha}{\alpha-1/2}\right)^{1/p} +1\right)r$ around $x$ contains $B(q)$, and therefore contains at least $\alpha w$ weight.
\end{proof}

The above algorithm does not work for any normed vector space.  Specifically, in Appendix \ref{MedianFail} we show that we do not always get an $O(1)$-approximation in the vector space of $\sqrt{d} \times \sqrt{d}$ matrices with operator norm distance.  However, we next present a more complicated algorithm that succeeds for any normed vector space.  It runs in $O_{\alpha}(nd)$ time for any normed vector space with fraction $\alpha > \frac{1}{2}$ and approximation constant $O((\alpha-1/2)^{-1})$, if distances and vector addition/scalar multiplication can be computed in $O(d)$ time, which is true for $\BR^d$ with an $\ell_p$ norm, or for $d^{\alpha} \times d^{\alpha}$-dimensional matrices with operator norm distances, where $\alpha$ is the inverse matrix multiplication constant, currently known to be at least $0.4214$ \cite{MatrixMultiplication}.

\begin{theorem} \label{EuclideanAboveHalf}
    For $\alpha > \frac{1}{2},$ in any normed vector space, if distances and addition/scalar multiplication of vectors can be calculated in $O(d)$ time, there exists an algorithm that solves the weighted problem in $O_{\alpha}(nd)$ time with fraction $\alpha$ and approximation constant $C = \frac{4\alpha}{2\alpha-1}.$
\end{theorem}

\begin{proof}
    If $n = 1$ we just return the first point so assume $n \ge 2.$  Given $n$ points, split the points into $n/2$ groups of $2$.  Assume $n$ is even, since if $n$ is odd, we can add a final point with $0$ weight.  Letting $m = \frac{n}{2},$ we construct balls $B_1, ..., B_m$, each of radius $2r$ as follows.  The ball $B_i$ will be centered around the point $a_{2i-1}$ or $a_{2i}$ with higher weight (we break ties with $a_{2i-1}$), so if $w_{2i-1} \ge w_{2i}$ we center around $a_{2i-1}$ and if $w_{2i-1} < w_{2i}$ we center around $a_{2i}$.
    
    Let $q_i$ be the center of $B_i,$ i.e. $q_i$ is either $a_{2i-1}$ or $a_{2i}$.  Let $B$ be a ball of radius $r$ containing points of total weight at least $\alpha w$, and let $q$ be the center of $B.$  
    
    We construct the new set of weights $v_i$ for the points $q_i$.  We let $v_i$ be the total $w$-weight of the subset of $\{a_{2i-1}, a_{2i}\}$ which is contained in $B_i$ minus the total $w$-weight of the subset which is not contained in $B_i.$  In other words, if $||a_{2i-1}-a_{2i}|| \le 2r,$ then $v_i = w_{2i-1}+w_{2i}$ and otherwise, $v_i = \max (w_{2i-1}, w_{2i}) - \min (w_{2i-1}, w_{2i}).$  Note that the total weight of $\{a_{2i-1}, a_{2i}\} \cap B_i$ is $\frac{w_{2i-1}+w_{2i}+v_i}{2}.$  Clearly, for all $i,$ $0 \le v_i \le w_{2i-1}+w_{2i}$.
    
    Next, if $||q_i-q|| > 3r,$ then $B_i$ and $B$ do not intersect.  This means that the total $w$-weight of $\{a_{2i-1}, a_{2i}\} \cap B$ is at most $\frac{w_{2i-1}+w_{2i}-v_i}{2}.$  If $||q_i-q|| \le 3r$, the total $w$-weight of the intersection $\{a_{2i-1}, a_{2i}\} \cap B$ is at most $\frac{w_{2i-1}+w_{2i}+v_i}{2},$ since if both $a_{2i-1}, a_{2i} \in B,$ then both are in $B_i,$ and if exactly one of $a_{2i-1}, a_{2i}$ is in $B$, then the one with larger weight is in $B_i$ because it is the center, $q_i$.  
    
    Now, define $S \subset [m]$ to be the set of $i$ such that $||q_i-q|| \le 3r,$ i.e. $S = \{i: 1 \le i \le m, ||q_i-q|| \le 3r\}.$  Then, by looking at the total $w$-weight of the subset of $a_{[n]}$ in $B$, 
    \[\sum\limits_{i \in S} \frac{w_{2i-1}+w_{2i}+v_i}{2} + \sum\limits_{i \not\in S} \frac{w_{2i-1}+w_{2i}-v_i}{2} \ge \sum\limits_{a_i \in B} w_i \ge \alpha w.\]
    Since $w$ is nonzero and $\alpha > \frac{1}{2}$, at least one $v_i$ is nonzero.  The left hand side equals 
    \[\frac{w}{2} + \frac{1}{2} \sum\limits_{i \in S} v_i - \frac{1}{2} \sum\limits_{i \not\in S} v_i,\]
    which means 
    \[\sum\limits_{i \in S}v_i - \sum\limits_{i \not\in S} v_i \ge (2\alpha-1)w \ge (2\alpha-1)\sum\limits_{1 \le i \le m} v_i \Rightarrow \sum\limits_{i \in S} v_i \ge \alpha \sum\limits_{1 \le i \le m} v_i.\]
    
    Therefore, the ball of radius $3r$ around $q$ contains at least $\alpha$ of the total $v$-weight of the points $q_i$.  Since at least one of the $v_i$'s is nonzero and all are nonnegative, we can find a ball of radius $3Cr$ around some point $p$ containing at least $\alpha$ of the total $v$-weight by performing the same algorithm on a size $m$ set $q_1, ..., q_m.$  Therefore, the ball of radius $3r$ around $q$ intersects the ball of radius $3Cr$ around $p,$ as some $q_i$ must be in both balls, so the ball of radius $(3C+4)r$ around $p$ must contain $B$.  Given this, if we can get some ball of radius $Cr$ that contains $B$, we are done.  
    
    We do this via looking at centroids, where the weighted centroid of points $x_1, ..., x_m$ with weights $w_1, ..., w_m$ equals $\frac{w_1x_1+...+w_mx_m}{w_1+...+w_m}$.  Let $\epsilon = \alpha - \frac{1}{2}$ and choose some $K \ge 2 + \frac{1}{\epsilon}$.  Suppose we have found some point $a$ such that the ball of radius $Kr$ around $a$, denoted $B^K(a),$ contains $B$.  We look at the $w$-weighted centroid of all points $a_i \in B^K(a)$, which clearly takes $O(nd)$ time to calculate.  If we let $a_{S_1} = a_{[n]} \cap B,$ then $w_{S_1} \ge \alpha w$ so the sum of the $w$-weights of points in $B^K(a)\backslash B$ is at most $w(1-\alpha).$  Then, the distance between the weighted centroid of all $a_i \in B^K(a)$ and $q$ is at most 
    \[\frac{1}{w_{S_1} + \sum_{a_i \in B^K(a) \backslash B} w_i} \left(\sum\limits_{a_i \in B} ||q-a_i|| w_i + \sum_{a_i \in B^K(a) \backslash B} ||q-a_i|| w_i\right)\]
    \[\le \frac{1}{w_{S_1} + \sum_{a_i \in B^K(a) \backslash B} w_i} \left(r w_{S_1}  + (2K-1)r \sum_{a_i \in B^K(a) \backslash B} w_i \right)\]
    since $||q-a|| \le (K-1)r$ and $||a-a_i|| \le Kr$ for any $a_i \in B^K(a) \backslash B.$  But since $w_{S_1} \ge \alpha w$ and $\sum_{a_i \in B^K(a) \backslash B} w_i \le (1-\alpha)w,$ this is at most 
    \[\alpha r + (2K-1)(1-\alpha)r = (2K-1-2K\alpha+2\alpha)r = (2K - 1 - K - 2K\epsilon+1+2\epsilon)r = (K-2(K-1)\epsilon)r.\]
    However, since $K \ge 2 + \frac{1}{\epsilon},$ $2(K-1)\epsilon \ge K\epsilon+1,$ so this is at most $(K-K\epsilon-1)r.$  Therefore, the weighted centroid of all these points is at most $(K-K\epsilon-1)r,$ so the ball of radius $K(1-\epsilon)r$ around the weighted centroid contains $B$.  This gives us a slightly better range.  We can repeat this process starting with $K = 3C+4$ until we get $K \le C,$ assuming that $C = 2+\frac{1}{\epsilon} = \frac{4\alpha}{2\alpha-1}.$  As $3C+4 \le 5C,$ this process needs to be repeated at most $(\log 5)/(\log \frac{1}{1-\epsilon}) = O(\epsilon^{-1})$ times.
    
    With the exception of the recursion on $q_1, ..., q_m$ with weights $v_1, ..., v_m,$ everything else takes $O(nd)$ time, but we have to repeat the centroid algorithm multiple times, where the number of repetitions depends on $\alpha$.  Therefore, the total running time is $T(n) = O_\alpha(nd) + T(n/2),$ which means $T(n) = O_\alpha(nd),$ as desired.
\end{proof}

\section{Normed Vector Space Algorithms: $\alpha > 0$} \label{NormedVectorSpaces2}

While we were unable to solve the normed vector space 1-center clustering with outliers problem for all $\alpha > 0$ in $O_{\alpha}(nd)$ time, we were able to find a solution running in $O_{\alpha, k}(nd \log^{(k)} n) = O_{\alpha, k}(nd \log...\log(n))$ time.  We first show an $nd \text{ polylog } n$ time solution and explain how this can be used to solve the problem in $O_{\alpha, k}(nd \log^{(k)} n)$ time.

The following result is useful for both the normed vector space and arbitrary metric space versions, primarily for $0 < \alpha \le \frac{1}{2}$.  It is important for making sure that if we found a ball of radius $Cr$ containing $\alpha w$ weight or $\alpha n$ points, even if there are multiple disjoint balls with this property, we can find a few balls of radius $Cr$, of which any ball of radius $r$ containing at least $\alpha w$ weight or $\alpha n$ points is near one of the radius $Cr$ balls.

\begin{lemma} \label{lame}
    Suppose we are in some space where computing distances between two points can be done in $O(d)$ time.  Suppose that for some fixed $\alpha, C$ and for any $\beta \ge \alpha$, we can solve the weighted problem with fraction $\beta$ and approximation constant $C$ in time $O(nd f(n))$ (with the runtime constant independent of $\beta$).  Then, for any $\beta \ge \alpha,$ we can find at most $\beta^{-1}$ points $p_1, ..., p_{\ell}$ in $O(ndf(n) \lfloor \beta^{-1} \rfloor)$ time such that the ball of radius $Cr$ around each $p_i$ contains at least $\beta w$ weight and any ball of radius $r$ containing at least $\beta w$ total weight intersects at least one of the balls of radius $Cr.$ 
\end{lemma}

The proof of lemma \ref{lame} is not too difficult and is left in Appendix \ref{A}.

\begin{lemma} \label{EuclideanFirstLemmaBelowHalf}
    For any $0 < \alpha < 1,$ let $C = 2+\frac{2}{\alpha}$ and assume we are dealing with the weighted problem in a normed vector space (with $w > 0$), where distances and vector addition/scalar multiplication are calculable in $O(d)$ time.  Suppose there exists a ball $B$ of radius $r$ such that $B$ and the ball $B^{2C+3}$ concentric with $B$ but of radius $(2C+3)r$ satisfies
    \[\sum\limits_{a_i \in B} w_i \ge \left(\sum\limits_{a_i \in B^{2C+3} \backslash B} w_i\right) + \alpha w.\]
    Then, we will be able to find a set of at most $\frac{1}{\alpha}$ points $z_1, ..., z_{\ell}$ in $O_{\alpha}(n d (\log n)^{\lfloor \alpha^{-1} \rfloor})$ time such that the ball $B^C(z_i)$ of radius $Cr$ around each $z_i$ contains at least $\alpha w$ total weight, and at least one of the balls $B^C(z_i)$ intersects the ball $B$.
    
    Also, if there does not exist such a ball $B$, the algorithm will still succeed and satisfy the conditions (where the condition of $B$ intersecting at least one of $B^C(z_i)$ is true by default).
\end{lemma}

\begin{proof}
    Our proof inducts on $\lfloor \alpha^{-1} \rfloor.$  We show an $O(nd \log n)$-time algorithm for $\alpha > \frac{1}{2}$ and given an $O(nd (\log n)^{k-1})$-time algorithm for all $\alpha' > \frac{1}{k},$ we show an $O(nd (\log n)^{k})$-time algorithm for all $\alpha > \frac{1}{k+1}$.  This means that the big $O$ time constant may depend on $\lfloor \alpha^{-1} \rfloor.$

    Assume $n$ is a power of $2$, as we can add extra points of weight $0$.  Next, split up the points $a_1, ..., a_n$ into two groups $a_{[n/2]}$ and $a_{[n/2+1::n]}.$  Note that $B$ clearly still holds the same property for either the first half or second half of points, i.e. either
    \[\sum\limits_{\substack{a_i \in B \\ 1 \le i \le n/2}} w_i \ge \alpha w_{[n/2]} + \sum\limits_{\substack{a_i \in B^{2C+3} \backslash B \\ 1 \le i \le n/2}} w_i \text{\space \space or } \sum\limits_{\substack{a_i \in B \\ n/2+1 \le i \le n}} w_i \ge \alpha w_{[n/2+1::n]} + \sum\limits_{\substack{a_i \in B^{2C+3} \backslash B \\ n/2+1 \le i \le n}} w_i.\]

    The algorithm first recursively runs on the two halves $a_{[n/2]}$ and $a_{[n/2+1::n]}$ to get points $x_1, ..., x_r$ and $y_1, ..., y_s$ such that $r, s \le \frac{1}{\alpha}$ and there exists some point $z \in \{x_1, ..., x_r, y_1, ..., y_s\}$ such that the ball of radius $Cr$ around $z$ intersects $B$.  Therefore, $B^{C+2}(z),$ the ball of radius $(C+2)r$ around $z$, contains $B$ but is contained in $B^{2C+3}$.
    
    Suppose we could successfully guess such a point $z$.  Then, the weight of points in $a_{[n]} \cap B^{C+2}(z)$ is $\beta w$ for some $\beta \ge \alpha,$ and so the weight of points in $a_{[n]} \cap B$ is at least $\frac{\beta+\alpha}{2} w$ since $B^{C+2}(z) \subset B^{2C+3}$.  We can easily determine the set of points in $a_{[n]} \cap B^{C+2}(z)$ in $O(nd)$ time, and thus compute $\beta.$  Now, among the points in $a_{[n]} \cap B^{C+2}(z),$ at least $\frac{\beta+\alpha}{2\beta} \ge \frac{1+\alpha}{2}$ of the weight is contained in some ball of radius $r$, which means by Theorem \ref{EuclideanAboveHalf}, we can in $O_{\alpha}(nd)$ time find a ball of radius
    \[\frac{4\left(\frac{\beta+\alpha}{2\beta}\right)}{2\left(\frac{\beta+\alpha}{2\beta}\right)-1} \cdot r = \frac{2(\beta+\alpha)}{\beta+\alpha - \beta} \cdot r = \left(2 + \frac{2\beta}{\alpha}\right)r \le Cr\]
    containing at least $\frac{\beta+\alpha}{2\beta} \cdot \beta w \ge \alpha w$ weight.
    
    If $\alpha > \frac{1}{2},$ this means we have found a ball of radius $Cr$ with at least $\alpha w$ total weight.  It must also intersect $B$, because otherwise the total weight of all the points would be at least $2\alpha w > w$.  Therefore, we can recursively run the algorithm on the two halves, and then in $O(nd)$ time guess at most $2$ possibilities for $z$ to find a ball of radius $Cr$.  Therefore, this algorithm takes $T(n) = 2T(n/2)+O(nd) \Rightarrow T(n) = O(nd \log n)$ time.
    
    Suppose $\frac{1}{k+1} < \alpha \le \frac{1}{k}.$  Then, in $O_{\alpha}(nd)$ time, we can try each $z \in \{x_1, ..., y_s\}$ to get some ball of radius $Cr$ centered around $z_1 = z$ that contains at least $\alpha w$ weight.  If we find no such ball, then no such $B$ exists, so we return nothing.  Else, we find some ball around $z_1$.  In case the ball does not intersect $B$, we compute the total weight of points in $B^C(z_1),$ the ball of radius $Cr$ around $z_1.$  Define $\gamma$ so that the weight of points in $B^C(z_1)$ equals $\gamma w$, so clearly $\gamma \ge \alpha.$  Therefore, if $B^C(z_1)$ does not intersect $B,$ then if we remove these points, we have a subset $\{a_1', ..., a_m'\}$ of the original points with total weight $w' = (1-\gamma) w,$ which means that for the new set of points, $B$ satisfies
    \[\sum\limits_{a_i' \in B} w_i' = \sum\limits_{a_i \in B} w_i \ge \left(\sum\limits_{a_i \in B^{2C+3}\backslash B} w_i\right) + \alpha w \ge \left(\sum\limits_{a_i' \in B^{2C+3}\backslash B} w_i'\right) + \frac{\alpha}{1-\gamma} w'.\]
    
    Thus, by our induction hypothesis, in $O_{\alpha/(1-\gamma)}(nd (\log n)^{\lfloor (1-\gamma)/\alpha \rfloor}) = O_{\alpha}(nd (\log n)^{\lfloor \alpha^{-1} \rfloor - 1})$ time, we can find a set of at most $\frac{1-\gamma}{\alpha} \le \frac{1}{\alpha} - 1$ points $z_2, ..., z_{\ell}$ such that the balls of radius $Cr$ around each $z_i$ contains at least $\frac{\alpha}{1-\gamma} w' = \alpha w$ weight in the new set of points (and thus in the old set of points), and at least one of the balls of radius $Cr$ around some $z_i$ (possibly $z_1$) intersects $B$.
    
    Since we first recursively perform the algorithm on the two halves, the total runtime is $T(n) = 2 \cdot T(n/2) + O_{\alpha}(nd (\log n)^{\lfloor \alpha^{-1} \rfloor - 1})$ by our inductive hypothesis, so $T(n) = O_{\alpha}(nd(\log n)^{\lfloor \alpha^{-1} \rfloor})$.
\end{proof}

We use the previous result to find an $O(nd \text{ polylog } n)$ time solution.

\begin{lemma} \label{ndpolylogn}
    For any $0 < \alpha < 1,$ one can solve the weighted Euclidean problem with fraction $\alpha$ and some approximation constant $C = O_\alpha(1)$ in $O_{\alpha}(nd (\log n)^{\lfloor 2\alpha^{-1} \rfloor})$ time.
\end{lemma}

\begin{proof}
    Suppose $B$ is a ball of radius $r$ around $p$ with $\alpha w$ points and let $S \subset \BN \cup \{0\}$ be the set of nonnegative integers $s$ such that there is a ball of radius $\left(\frac{8}{\alpha}+7\right)^s \cdot r$ around $p$ containing at least $(\frac{3}{2})^s \cdot \alpha w$ total weight.  Because of $B$, $0 \in S$.  Since $\alpha > 0,$ there clearly exists a maximal $s \in S$ which is at most $\frac{\log (\alpha^{-1})}{\log (3/2)}.$  For this maximal $s,$ there is a ball $B'$ of radius $R = \left(\frac{8}{\alpha}+7\right)^s \cdot r$ around $p$ containing at least $\alpha' w$ weight, where $\alpha' = (\frac{3}{2})^s \alpha,$ but the ball of radius $\left(\frac{8}{\alpha} + 7\right) R$ around $p$ contains at most $\frac{3}{2} \alpha' w$ total weight.  Therefore, if $\beta = \frac{\alpha}{2},$ if we let $C = 2 + \frac{2}{\beta},$ the ball $(B')^{2C+3}$ of radius $(2C+3)R = \left(\frac{8}{\alpha} + 7\right) R$ around $p$ satisfies
    \[\sum\limits_{a_i \in B'} w_i \ge \left(\sum\limits_{a_i \in (B')^{2C+3} \backslash B'} w_i\right) + \beta w.\]
    Therefore, if we knew $s,$ plugging $\beta$ into the algorithm of Lemma \ref{EuclideanFirstLemmaBelowHalf} gives us, in $O_{\alpha}(nd(\log n^{\lfloor 2\alpha^{-1} \rfloor}))$ time, at most $2 \alpha^{-1}$ points such that the ball of radius $\left(\frac{4}{\alpha}+2\right) \cdot \left(\frac{8}{\alpha}+7\right)^s$ around at least one of them intersects $B',$ and thus the ball of radius $\left(\frac{4}{\alpha}+4\right) \cdot \left(\frac{8}{\alpha}+7\right)^s$ around that point has at least $\alpha w$ weight.  We can try it for all $s$ between $0$ and $\frac{\log (\alpha^{-1})}{\log (3/2)}$ and verify each point (verification takes $O_{\alpha}(nd)$ time) to get at least one ball containing $\alpha w$ or more weight, which gives the desired result.
\end{proof}

\vskip -0.15cm

We now can go to $O_{\alpha, k}(nd \log^{(k)}(n))$ time using the following lemma.

\begin{lemma} \label{BruteForceDivConquer}
    Fix some $\alpha, C$ and suppose we are in some space (Euclidean, general metric, or something else) where distances can be computed in $O(d)$ time.  Suppose that for any fraction $\beta \ge \alpha$ and approximation constant $C$ there exists an algorithm that solves the weighted problem in time $O(ndf(n))$.  Then, for any nondecreasing function $g(n)$ such that $1 \le g(n) \le n,$ there is an algorithm that runs in $O\left(ndf(g(n)) + \frac{ndf(n)}{g(n)}\right)$ with fraction $\alpha' = \sqrt{2\alpha}$ and approximation constant $C' = C^2+2C+2$.
\end{lemma}

\begin{proof}
    
    We use a similar divide and conquer approach to Lemma \ref{EuclideanFirstLemmaBelowHalf}.  Partition $[n]$ into buckets $D_1, ..., D_m,$ each of size $\Theta(g(n)),$ which gives us a partition of points $a_{D_1}, ..., a_{D_m}.$  If $B$ is a ball of radius $r$ containing at least $\alpha' w$ total weight, then let $v_i$ be the total weight of all points in $a_{D_i} \cap B.$  If $S \subset [m]$ is the set of all $i$ such that $v_i > \frac{\alpha'}{2} w_{D_i},$ then
    \[\alpha' w \le \sum\limits_{a_j \in B} w_j = \sum\limits_{i \in [m]} \sum\limits_{\substack{j \in D_i \\ a_j \in B}} w_j \le \sum\limits_{i \in S} w_{D_i} + \frac{\alpha'}{2} \sum\limits_{i \not\in S} w_{D_i} \le \frac{\alpha' w}{2} + \sum\limits_{i \in S} w_{D_i},\]
    and thus $\frac{\alpha'}{2} w \le \sum_{i \in S} w_{D_i}.$

    For each $1 \le i \le m,$ by Lemma \ref{lame}, since $\alpha' \ge \alpha$, there is an $O(nd f(g(n)))$-time algorithm which returns for each $i \in [m] $ at most $\alpha'^{-1}$ points $p_{i, 1}, ..., p_{i, \ell_i}$ such that if $i \in S$, the ball of radius $Cr$ around at least one of the points intersects $B.$  Therefore, for every $i \in S$, some $p_{i, j}$ is at most $(C+1)r$ from the center of $B.$  Now, we can compute $w_{D_1}, ..., w_{D_m}$ in $O(n)$ time and assign each $p_{i, j}$ weight $w_{D_i}.$  Then, the total weight of all $p_{i, j}$ is at most $\alpha'^{-1} w$.  However, for an individual $i \in S,$ the total weight of the points $p_{i, j}$ for all $1 \le j \le \ell_i$ in the ball of radius $(C+1)r$ around $B$ is at least $w_{D_i}$ since at least one $p_{i, j}$ is in the ball.  Therefore, the total weight of all points $p_{i, j}$ in the ball of radius $(C+1)r$ around $B$ is at least $\sum_{i \in S} w_{D_i} \ge \frac{\alpha'}{2} w,$ which is at least $\frac{\alpha'^2}{2}$ times the total weight of all the $p_{i, j}$’s.
    Therefore, by Lemma \ref{lame}, applying the algorithm for $\alpha = \frac{\alpha'^2}{2}$ on the $p_{i, j}$'s with the new radius $(C+1)r$ gives a set of at most $\alpha^{-1}$ points $q_1, ..., q_{\ell}$ such that the ball of radius $C(C+1)r$ around at least one of the $q_i$'s intersects the ball of radius $(C+1)r$ around the center of $B.$  This algorithm takes $O(\alpha^{-1} md f(m)) = O(nd \frac{f(n)}{g(n)})$ time, as $\alpha$ is fixed.  Therefore, the ball of radius $(C^2+2C+2)r = C'r$ around at least one of the $q_i$’s contains $B,$ so we verify for each $q_i$ if the ball of radius $(C^2+2C+2)r$ contains at least $\alpha w$ total weight, which takes $O(nd)$ time.
\end{proof}

\begin{theorem} \label{logtower}
    For $\alpha \le \frac{1}{2},$ the 1-center clustering with outliers problem can be solved in $O_{\alpha}(nd \log^{(k)}(n))$ time in any normed vector space for some constant $C = O_{\alpha}(1).$
\end{theorem}

\begin{proof}
    Letting $f = g$ in Lemma \ref{BruteForceDivConquer} tells us there is an $O(ndf(f(n))$-time algorithm with fraction $\sqrt{2\alpha}$ and approximation constant $C^2+2C+2$ given an $O(nd f(n))$-time algorithm with fraction $\alpha$ and approximation constant $C$.  Repeating this $k$ times gives us an $O_k(nd f^{(2^k)}(n))$-time algorithm with fraction $2 \cdot (\alpha/2)^{2^{-k}}$ and approximation constant $O_{C, k}(1).$  Therefore, since we have an algorithm running in $O_{\alpha}(ndf(n))$ with $f(n) = (\log n)^{\lfloor 2 \alpha^{-1} \rfloor}$ with approximation constant $O_{\alpha}(1)$ and fraction $\alpha,$ we have an algorithm that runs in $O_{\alpha, k}(ndf^{(2^k)}(n)) = O_{\alpha, k}(nd \log^{(2^k-1)} n)$ time, with fraction $2 \cdot (\alpha/2)^{2^{-k}}$ and approximation constant $O_{\alpha, k}(1).$  Letting $\beta = 2 \cdot (\alpha/2)^{2^{-k}},$ then $\alpha = (\beta/2)^{2^k}/2,$ which means for any $0 < \beta < 1,$ there is an $O_{\beta, k}(nd \log^{(2^k-1)}(n))$ time solution with approximation constant $O_{\beta, k}(1)$ and fraction $\beta$.
\end{proof}

\section{Metric Space Upper Bounds} \label{SectionMetricUpper}

The idea for proving that there is an $O_{\alpha, C}(n^{1+1/C})$-time algorithm with fraction $\alpha$ and approximation constant $2C$ uses induction on $\lfloor \alpha^{-1} \rfloor$ and $C$.  The base case proofs of $\alpha > \frac{1}{2}$ and $C = 1$ are quite similar to the induction step, so we leave their proofs in Appendix \ref{A}.

\begin{theorem} \label{MetricBelowHalf}
    For any $\alpha > 0,$ say we are trying to solve weighted 1-center clustering with outliers in a general metric space, where $r$ is unknown.  For all $C \in \mathbb{N},$ we can find a set of points $p_1, ..., p_\ell$ and corresponding radii $s_1, ..., s_{\ell},$ where $\ell \le \lfloor \alpha^{-1} \rfloor$, such that the ball of radius $s_i$ around $p_i$ contains at least $\alpha w$ of the weight in $O((2{\lfloor \alpha^{-1} \rfloor + C \choose C}-\lfloor \alpha^{-1} \rfloor-1) n^{1+1/C})$ time, assuming $n = m^C$ for some integer $m$.  Moreover, any ball of radius $r$ containing at least $\alpha w$ weight intersects at least one ball of radius $s_i$ around some $p_i$, for some $s_i \le 2Cr$.
\end{theorem}

\begin{proof}
    We induct on $\lfloor \alpha^{-1} \rfloor$ and $C$.  The base cases $\lfloor \alpha^{-1} \rfloor = 1$ and $C = 1$ are done in Appendix \ref{A}.  Suppose the theorem holds for all $\alpha' > \frac{1}{z}$ and we are looking at some $\frac{1}{z+1} < \alpha \le \frac{1}{z}.$  Also, suppose we have an algorithm for $\alpha$ and $C-1.$
    
    Split the points into blocks $D_1, ..., D_m$ each of size $m^{C-1}.$  For each block $D_i,$ by our inductive hypothesis we can return points $p_{i, 1}, ..., p_{i, \ell_i}$ and radii $r_{i_1}, ..., r_{i, \ell_i} \in a_{D_i}$ where $\ell_i \le z$ for all $i$, subject to some conditions.  First, the ball $B_{i, k}$ of radius $r_{i, k}$ around $p_{i, k}$ has at least $\alpha w_{D_i}$ weight.  Second, if there is a ball of radius $r$ that contains at least $\alpha w_{D_i}$ weight when intersected with $a_{D_i}$, then the ball must intersect $B_{i, k}$ for some $k$ where $p_{i, k} \le 2(C-1)r$.  Moreover, by our induction hypothesis we can determine these points in time
    \[\left(\left(2{z + C-1 \choose C-1} - z - 1 \right)\left(\frac{n}{m}\right)^{1+1/(C-1)} \cdot m\right) = O\left(\left(2{z + C-1 \choose C-1} -z-1 \right)n^{1+1/C}\right).\]
    
    If $B$ is a ball of radius $r$ containing at least $\alpha w$ total weight, then there exists some $1 \le j \le m$ such that $w_{D_j} > 0$ and the total weight of $a_{D_j} \cap B$ is at least $\alpha w_{D_j}.$  Therefore, $B_{j, k}$ intersects $B$ for some $r_{j, k} \le 2(C-1)r$, so the ball of radius $2Cr$ around $p_{j, k}$ for some $j, k$ when intersected with $a_{[n]}$ contains at least $\alpha w$ total weight.  We can check all the $p_{j, k}$ and since weighted median can be solved in $O(n)$ time, we can find some $p_{j, k}$ with the smallest radius $s_{j, k}$ (not necessarily the same as $r_{j, k}$) containing at least $\alpha w$ weight in $O(mz \cdot n) = O(z n^{1+1/C}).$  We know that $s_{j, k} \le 2Cr,$ and we can set $p_1 = p_{j, k}$ and $s_1 = s_{j, k}$.
    
    Now, remove every point in the ball of radius $s_1$ around $p_1$ by changing their weights to $0$.  If the total weight of removed points is $\beta w$ where $\beta \ge \alpha,$ the total weight is now $(1-\beta)w.$  If there is still some ball of radius $r$ that contains at least $\alpha w$ weight now, then it contains at least $\frac{\alpha}{1-\beta} > \frac{1}{z-1}$ of the total weight now.  Therefore, we can use induction on $z$ with $\alpha' = \frac{\alpha}{1-\beta}$.  This gives us at most $z$ points $p_1, ..., p_\ell$ and radii $s_1, ..., s_{\ell},$ where the first point $p_1$ is our original $p_{j, k}$ and the next $\ell-1$ points and radii are found in $O\left(\left(2{z-1 + C \choose C}-(z-1)-1\right) n^{1+1/C}\right)$ time.  Moreover, 
    any ball $B$ of radius $r$ either intersects the ball of radius $s_1$ around $p_1,$ where $s_1 \le 2Cr,$ or by the induction hypothesis on $\lfloor \alpha^{-1} \rfloor$ intersects some $s_i$ around $p_i$ for some $2 \le i \le \ell$ with $s_i \le 2Cr,$ since $B$ would have at least $\frac{\alpha}{1-\beta}$ of the remaining weight if it doesn't intersect the ball of radius $s_1$ around $p_1$.
    
    Therefore, the total time is 
    \[O\left(\left(2{z + C-1 \choose C-1}-z-1\right)n^{1+1/C} + \left(2{z-1 + C \choose C}-z\right) n^{1+1/C} + zn^{1+1/C}\right)\] 
    \[= O\left(\left(2{z+C \choose C} - z - 1\right) n^{1+1/C}\right).\]
\end{proof}

\begin{remark}
    $C$ does not have to be a constant independent of $n$, since the $O$ factor is independent of $z$ and $C$.  For example, $m = 2, C = \lceil \lg n \rceil$, the theorem still holds.
\end{remark}

As we can add points of $0$ weight until we get a perfect power of $C$, we have the following.

\begin{corollary}
    In any metric space, we can find a ball of radius $2Cr$ with at least $\alpha n$ points in $O_{\alpha, C}(n^{1+1/C})$ time, given that there exists a ball of radius $r$ with at least $\alpha n$ points.
\end{corollary}

\section{Metric Space Lower Bounds} \label{lowerbounds}

The goal of this section is to prove the following theorem.

\begin{theorem} \label{MainLower}
    For any fixed $\beta > 0,$ there exists some constant $\epsilon > 0$ depending on $C, \alpha, \beta$ such that there is no algorithm taking fewer than $\epsilon \cdot n^{1+1/(C-1)}$ queries that deterministically finds a $2C(1-\beta)$ approximation to $1$-center clustering with outliers with fraction $\alpha$.
\end{theorem}

Note that the following is a direct corollary:

\begin{theorem}
    For any fixed $0 < \epsilon < 1,$ and any fixed $\alpha,$ the smallest $c$ such that 1-Center Clustering with Outliers has a $c$-approximation with fraction $\alpha$ in $O(n^{1+\epsilon})$ time, or even in $O(n^{1+\epsilon})$ queries, equals $2 \left\lfloor \frac{1}{\epsilon}\right\rfloor.$  Consequently, there is no $o(n^{1+1/C})$-time or even $o(n^{1+1/C})$-query algorithm which provides a $2C$-approximation to $1$-Center Clustering with Outliers in an arbitrary metric space.
\end{theorem}

To prove Theorem \ref{MainLower}, we use Chang's adversary from \cite{Lowerbound1}. For some fixed constant $C \in \BN \backslash \{1\}$, assume that we have an algorithm $A$ giving a $2C(1-\beta)$ approximation with $q = o(n^{1+1/(C-1)})$ queries.  By this, we mean $q \le \epsilon \cdot n^{1 + 1/(C-1)},$ where $\epsilon = \epsilon(C, \delta) > 0$ and $\delta = \delta(\alpha, \beta) > 0$ are some small constants such that $\delta$ is fixed for some fixed $\alpha, \beta$ (recall that $\alpha$ is our fraction in the 1-Center Clustering with Outliers problem) and $\epsilon$ is fixed for fixed $C, \delta$.  Assume $n$ is sufficiently large and WLOG that $q$ is at least $n$, by padding $n$ dummy queries at the end if needed.  Chang provided an adversary that acts against any algorithm $A$ which asked $q$ queries $(i_1, j_1), ..., (i_q, j_q)$ where each $s$th query requested the distance between $a_{i_s}$ and $a_{j_s}$ and finally returned some point $a_z.$  The adversary returns after each query some nonnegative real.  Moreover, for sufficiently large $n$, the values the adversary returns are consistent with the distances between $a_1, ..., a_n$ being in a metric space \cite[Lemma 3.28]{Lowerbound1}, and the following are true for sufficiently small $\epsilon$:
\begin{enumerate}
    \item For the returned point $a_z,$
\[\sum\limits_{i \in [n]} \rho(a_z, a_i) \ge n \cdot \left(C - \delta \right) \]
    \item There exists $y \in [n]$ such that
\[\sum\limits_{i \in [n]} \rho(a_y, a_i) \le n \cdot \left(\frac{1}{2} + \delta \right) \]
    \item For all $i, j,$ $\rho(a_i, a_j) \ge \frac{1}{2}.$
    \item For all $i, j,$ $\rho(a_i, a_j) \le C.$
\end{enumerate}
The first follows from \cite[Equation 55]{Lowerbound1} and the (unlabeled) equation right above it in \cite{Lowerbound1}, the second follows from \cite[Equation 54]{Lowerbound1}, and the third follows from \cite[Lemma 3.27]{Lowerbound1}, and the fourth follows from \cite[Equation 48]{Lowerbound1} and \cite[Equation 49]{Lowerbound1}.

Using the above, we can prove Theorem \ref{MainLower}.

\begin{proof}
    For any algorithm $A$ taking $o(n^{1+1/(C-1)})$ time, for sufficiently large $n$ the algorithm will return a point $a_z$ such that $\sum \rho(a_z, a_i) \ge n(C-\delta).$  This means that if a ball of radius $K_1$ around $a_z$ contains at least $\alpha n$ points, 
\[C \cdot (1-\alpha) n + K_1 \cdot \alpha n \ge \sum_i \rho(a_z, a_i) \ge n \cdot (C - \delta)\]
    so $C(1-\alpha)+K_1 \alpha \ge C-\delta,$ which means $K_1 \alpha \ge C \alpha - \delta$ so $K_1 \ge C - \frac{\delta}{\alpha}$.  However, if the ball of radius $K_2$ around $a_y$ is the smallest ball centered at $a_y$ containing at least $\alpha n$ points, then 
\[\frac{1}{2} \cdot (\alpha n - 1) + K_2 \cdot ((1-\alpha) n) \le \sum_i \rho(a_y, a_i) \le n \cdot \left(\frac{1}{2} + \delta\right).\]  
    This means that 
\[K_2(1-\alpha) \le \frac{1}{2} + \delta - \frac{1}{2} \alpha + \frac{1}{2n} = \frac{1}{2} (1-\alpha) + \delta + \frac{1}{2n},\]
    so $K_2 \le \frac{1}{2} + \frac{\delta + 1/(2n)}{1-\alpha}.$
    
    Clearly, 
\[\frac{C - \frac{\delta}{\alpha}}{\frac{1}{2} + \frac{\delta + 1/(2n)}{1-\alpha}} \ge \frac{C - \frac{\delta}{\alpha}}{\frac{1}{2} + \frac{2\delta}{1-\alpha}} \ge 2C(1-\beta)\]
    is true for sufficiently large $n$ and some fixed $\delta > 0$, as the first inequality is true as long as $\frac{1}{2n} < \delta,$ and the middle term is clearly continuous for $\delta \ge 0$ and converges to $2C$ as $\delta$ goes to $0$, so the second inequality is true for all sufficiently small $\delta$.
    
    Therefore, there exists $\delta = \delta(\alpha, \beta) > 0$ and $\epsilon = \epsilon(C, \delta) > 0$ such that no algorithm with $\epsilon \cdot n^{1+1/(C-1)}$ queries can be a deterministic $2C(1-\beta)$-approximation algorithm to 1-center clustering with outliers.
\end{proof}

\section*{Acknowledgements}
I want to thank Professors Piotr Indyk and Jelani Nelson, who proposed this problem in a course I took with them.  I would also like to thank them both for providing a lot of feedback on my work and write-up.  I would also like to thank James Tao for a helpful discussion.

\newpage
\appendix

\section{Omitted Proofs} \label{A}

First, we prove Lemma \ref{lame} from Section \ref{NormedVectorSpaces2}.

\begin{proof}[Proof of Lemma \ref{lame}]
    We induct on $\lfloor \alpha^{-1} \rfloor.$  For $\alpha > \frac{1}{2}$ and some $\beta \ge \alpha,$ we can in $O(ndf(n))$ time output $p$ such that the ball of radius $Cr$ around $p$ contains at least $\beta w$ total weight.  But then since $\beta > \frac{1}{2},$ the second condition is true by default, so we are done.  Also, if there is no ball of radius $r$ containing $\alpha w$ weight, our algorithm may output some point, but in $O(nd)$ time we can verify and either output a ball of radius $Cr$ containing $\alpha w$ weight, or output nothing. 
    
    Suppose $\alpha > \frac{1}{z+1}$ and we know it is true for all $\alpha' > \frac{1}{z}.$  In $O_{\alpha}(ndf(n))$ time, we can find $B^C(p_1),$ a ball of radius $Cr$ around some $p_1$ containing at least $\beta w$ total weight for some $\beta \ge \alpha$.  Again, if no such ball of radius $r$ exists, we will either get nothing, in which case we end the program, or may happen to get a point $p_1$ such that $B^C(p_1)$ contains $\alpha w$ weight.  Assuming we got a point in $O(nd)$ time we can remove all points in $B^C(p_1)$ by just checking all points' distances from $p_1.$  Then, the remaining weight is $(1-\beta') w$ for some $\beta' \ge \beta,$ and $\beta'$ can be calculated in $O(nd)$ time.

    If there exists a ball $B$ of radius $r$ that doesn't intersect $B^C(p_1),$ none of the points in $B$ were removed, which means it has at least $\frac{\beta}{1-\beta'}$ of the remaining weight.  Let $B^C(p_i)$ be the ball of radius $Cr$ around $p_i.$  We apply the induction hypothesis with fraction $\frac{\beta}{1-\beta'} > \frac{1}{z}.$  It tells us in $O(ndf(n) (z-1))$ time we can find at most $z-1$ points $p_2,..., p_{\ell}$ such that every ball of radius $r$ containing at least $\alpha w$ weight either intersects $B^C(p_1)$ or it still contains at least $\frac{\beta}{1-\beta'}$ of the remaining weight, which means it intersects $B^C(p_i)$ for some $2 \le i \le \ell$.

    If there does not exist a ball of radius $r$ containing at least $\alpha w$ weight not intersecting $B^C(p_1)$, we will either output no points after $p_1,$ or we may still output some points $p_2, ..., p_\ell$ such that $B^C(p_i)$ contains at least $\frac{\beta}{1-\beta'}$ of the remaining weight, or $\alpha w$ total weight.  But since every ball of radius $r$ containing at least $\alpha w$ weight intersects $B^C(p_1),$ we are done.
\end{proof}

Next, we prove the base cases of $\lfloor \alpha^{-1} \rfloor = 1$ and $C = 1$ for Theorem \ref{MetricBelowHalf}.

\begin{theorem} \label{MetricAboveHalf}
    For $\alpha > \frac{1}{2},$ suppose we are trying to solve the weighted $1$-center clustering problem in a general metric space, but now assuming $r$ is unknown.  Then, for any positive integer $C$, we can find a point $p$ such that the ball of radius $2Cr$ around $p$ contains at least $\alpha w$ of the weight in $O(Cn^{1+1/C})$ time, assuming $n = m^C$ for some integer $m$.  As an obvious consequence, every ball of radius $r$ containing at least $\alpha w$ of the weight must intersect the ball of radius $2Cr$ around $p$.
\end{theorem}

\begin{proof}
    For $C = 1,$ we compute for each $a_i$ the quantities $\rho(a_i, a_1), ..., \rho(a_i, a_n)$ and let $r_i$ be the smallest real number such that the ball of radius $r_i$ around $a_i$ contains at least $\alpha w$ total weight.  This can be computed for each $i$ in $O(n)$ time using standard algorithms for weighted median, and thus takes a total of $O(n^2)$ time for all $i$.  Then, if some $r_i$ equals $\min(r_1, ..., r_n),$ the ball of radius $r_i$ around $a_i$ contains at least $\alpha w$ total weight, and $r_i \le 2r$ since otherwise, there is a ball of radius $r$ around some $p$ in the metric space containing at least $\alpha w$ total weight, which means the ball of radius $2r$ around around some $p_j$ in that radius $r$ ball must contain at least $\alpha w$ total weight, so $r_i \le 2r.$  This proves our claim for $C = 1.$
    
    Assume there is an algorithm that works for $C-1.$  Then, split the $n$ points into $m$ blocks $D_1, ..., D_m$ of size $m^{C-1}.$  For each block $D_i,$ we can return $p_i \in a_{D_i}$ such that if there is a ball of radius $r$ that when intersected with $a_{D_i}$ contains at least $\alpha w_{D_i}$ weight, then the ball of radius $2(C-1)r$ around $p_i$ intersected with $a_{D_i}$ contains at least $\alpha w_{D_i}$ weight.  Moreover, we can determine $p_1, ..., p_m$ in $O((C-1)(n/m)^{1+1/(C-1)} \cdot m) = O((C-1)n^{1+1/C})$ time.  
    
    If $B$ is a ball of radius $r$ containing at least $\alpha$ of the total weight, then there exists some $1 \le k \le m$ such that $w_{D_k} > 0$ and the total weight of $a_{D_k} \cap B$ is at least $\alpha w_{D_k}.$  Since the ball of radius $(2C-2)r$ around $p_k$ contains at least $\alpha w_{D_k}$ weight when intersected with $a_{D_k},$ and since $\alpha > \frac{1}{2},$ the ball of radius $(2C-2)r$ around $p_k$ must intersect $B$.  Therefore, the ball of radius $2Cr$ around $p_k$ contains $B$ and thus contains at least $\alpha w$ weight when intersected with $a_{[n]}.$
    
    This means after we get our points $p_1, ..., p_m,$ the ball of radius $2Cr$ around at least one of the $p_i$'s must have at least $\alpha w$ total weight.  We determine $r_1, ..., r_m$ where $r_i$ is the radius of the smallest ball around $p_i$ containing at least $\alpha w$ of the original weight, which can be done in $O(n)$ time for each $i$ since weighted median can be solved in $O(n)$ time.  Doing this for each $p_i$ takes $O(nm) = O(n^{1+1/C})$ time, and if $r_i = min(r_1, ..., r_m)$ for some $i$, then clearly $r_i \le 2Cr$.  Therefore, this takes $O((C-1)n^{1+1/C}) + O(n^{1+1/C}) = O(Cn^{1+1/C})$ time total, so our induction step is complete.
\end{proof}

\begin{lemma} \label{MetricBelowHalfEasy}
    For any $\alpha > 0,$ say we are trying to solve weighted 1-center clustering with outliers in a general metric space, with $r$ unknown.  In $O(\alpha^{-1} n^2)$ time we can find $\ell \le \lfloor \alpha^{-1} \rfloor$ points $p_1, ..., p_{\ell}$ with corresponding radii $s_1, ..., s_{\ell}$ such that the ball of radius $s_i$ around $p_i$ contains at least $\alpha w$ weight.  Moreover, any ball of radius $r$ containing at least $\alpha w$ weight will intersect at least one ball of radius $s_i$ around $p_i$ where $s_i \le 2r$.
\end{lemma}

\begin{proof}
    Define $y = \alpha w.$  Like in Theorem \ref{MetricAboveHalf}, we find for each $a_1, ..., a_n$ values $r_1, ..., r_n$ such that $r_i$ is the smallest radius around $a_i$ of a ball containing at least $\alpha w$ total weight, and these can all be done in $O(n^2)$ time.  Let $p_1$ be the point $a_i$ with smallest corresponding $r_i$, and let $s_1$ be the corresponding $r_i$.  Clearly, $r_i \le 2r$ and the total weight of the points in the ball of radius $r_i$ around $p_1$ is at least $y$.  Remove all the points in this ball.  Repeat this procedure (for the same $y$, not $\alpha$ times the new total weight) until we have $p_1, ..., p_{\ell}$ and the remaining points have weight less than $y$.  This procedure clearly takes $O(\alpha^{-1} n^2)$ time.  
    
    Suppose some ball $B$ contains at least $\alpha w$ weight but does not intersect a ball of radius $s_i$ around $r_i$ for any $i$ such that $s_i \le 2r.$  Then, suppose $j$ is the largest integer such that $s_i \le 2r$ for all $i \le j.$  Either $j = \ell$ or $s_{j+1} > 2r.$  If $j = \ell,$ then the remaining points have weight less than $y$, which makes no sense since $B$ has weight at least $y$ and does not intersect any of the balls we created.  If $s_{j+1} > 2r,$ we would have picked a different ball.  This is because if $a_k \in B,$ the ball of radius $2r$ around $a_k$ contains at least $\alpha w$ weight, so we would have picked $a_k$ as our point $p_{j+1}$ instead.  Thus, we are done.
\end{proof}

\section{A Normed Vector Space where the ``Median of Each Coordinate'' Algorithm Fails} \label{MedianFail}

Here, we note that the simple median algorithm described in Theorem \ref{MedianOfEachCoordinate} does not work in certain normed vector spaces even when coordinates are given.  Specifically, we note that this algorithm fails for square matrices with operator norm distance.

Consider the normed vector space of $\sqrt{d} \times \sqrt{d}$-dimensional matrices with the operator norm distance.  Distances and vector space operations can easily be done in $O(d)$ time, and therefore, Theorem \ref{EuclideanAboveHalf} can be used to provide an $O(nd)$ time algorithm to solve $1$-center clustering with outliers when $\alpha > \frac{1}{2}$.  However, while we can compute the median of each coordinate in $O(nd)$ time, this algorithm fails to be an $O(1)$-approximation.  We formally state and prove it now:

\begin{theorem}
    Let $d$ be a perfect square.  There exist $n = 2^d-1$ matrices $a_1, ..., a_n \in \BR^{\sqrt{d} \times \sqrt{d}}$ such that $n-o(n)$ of the points $a_1, ..., a_n$ are in the operator norm ball centered around the origin with radius $d^{1/3},$ but the algorithm from Theorem \ref{MedianOfEachCoordinate} returns a point with operator norm $\sqrt{d}$.  Consequently, the algorithm from Theorem \ref{MedianOfEachCoordinate} cannot work for this vector space for any $0 < \alpha < 1$.
\end{theorem}

\begin{proof}
    Define $k = \sqrt{d}.$  Consider the set of all $k \times k$ matrices with each coordinate $\pm 1,$ with the exception of the matrix with all coordinates $-1$.  There are clearly $2^d-1$ such matrices, and it is easy to see that the median of each coordinate equals $1$.
    
    The operator norm of the matrix with all $1$'s equals $k = \sqrt{d}$.  However, it is known that if a matrix $M \in \BR^{k \times k}$ is chosen with i.i.d. entries from $\pm 1,$ the probability of its operator norm being more than $k^{2/3}$ is $O(e^{-\Omega(k^{7/6})}) = o(1)$ (see for example \cite[Corollary 2.3.5]{Tao}).
    Since our set $a_1, ..., a_n$ contains every single matrix with entries in $\pm 1$ with equal probability, except the matrix of all $1$'s, we have that $1-o(1)$ fraction of the $a_i$'s have operator norm less than $k^{2/3} = d^{1/3},$ which is precisely what we want.
    
    The last statement of the theorem is now immediate, since the algorithm returns a point with operator norm $d^{1/2} = \omega(d^{1/3})$, so it does not provide an $O(1)$-approximation.
\end{proof}

\end{document}